\documentclass{amsart}

\usepackage{bm}
\usepackage{graphicx}
\usepackage{amsmath}
\usepackage{hyperref}
\usepackage{wrapfig}
\usepackage{amsfonts}
\usepackage{amsthm}

\numberwithin{equation}{section}
\newtheorem{thm}{Theorem}[section]

\newtheorem{lemma}[thm]{Lemma}
\newtheorem{claim}{Claim}

\newtheorem{remark}{Remark}
\newtheorem{definition}{Definition}

\newcommand{\p}{\partial}

\newcommand{\be}{\begin{equation}}
\newcommand{\ee}{\end{equation}}
\newcommand{\la}{\label}

\begin{document}

\title[Topological constraints in geometric deformation quantization]{Topological constraints in geometric deformation quantization on domains with multiple boundary components}

%

%
%
\author[Razvan Teodorescu]{Razvan Teodorescu}
\email{razvan@usf.edu}
\address{4204 E. Fowler Ave., CMC 314, Tampa, FL 33620}


\begin{abstract}
A topological constraint on the possible values of the universal quantization parameter is revealed in the case of geometric quantization on (boundary) curves diffeomorphic to $S^1$, analytically extended on a bounded domain in $\mathbb{C}$, with $n \ge 2$ boundary components. Unlike the case of one boundary component (such as the canonical Berezin quantization of the Poincar\'e upper-half plane or the case of conformally-invariant 2D systems), the more general case considered here leads to a strictly positive minimum value for the quantization parameter, which depends on the geometrical data of the domain (specifically, the total area and total perimeter in the smooth case).  It is proven that if the lower bound is attained, then $n=2$ and the domain must be annular, with a direct interpretation in terms of the global monodromy. 
\end{abstract}

\maketitle

\section{Introduction}

Deformation quantization of low-dimensional physical systems characterized by non-trivial topological and geometric data is a problem found at the core of several open question of interest for current research, especially those related to Fractional Quantum Hall systems, non-equilibrium quantum dynamics of many-body states in cold Fermi gases, or strongly-interacting Fermionic systems in 1+1 dimensions, to name only a few. The problem investigated in this article can be summarized as follows: given a bounded domain $\Omega \subset \mathbb{C}$, of total area $A$ (or the corresponding Minkowski content with respect to Lebesgue measure, in a more general case), and $n \in \mathbb{N}, n \ge 2$ boundary components 
$\Gamma_k, k = 1, 2, \ldots, n$, with total perimeter $P = \sum_{k=1}^n L_k$ (where $L_k$ is the perimeter of $\Gamma_k$, or corresponding Minkowski content relative to the arclength Lebesgue measure), are there any constraints on the possible values of the deformation quantization parameter $\lambda$, imposed by the topological and geometrical data of the domain? 
 
It is important to stress from the beginning that, aside from the two-dimensional set-up (and corresponding assumptions on the Minkowski content of the domain), the question is addressed at the highest level of generality; specifically, the only non-trivial assumption made throughout this study is that the quantization parameter $\lambda$ is the same for each boundary component (and indeed for the whole domain $\Omega$), in the sense that we do not allow for various ``boundary" restrictions of the theory (corresponding to -- in the case of smooth boundaries diffeomorphic to the unit circle $S^1$ -- various realizations of the Virasoro algebra deformations of orientation-preserving diffeomorphisms of $S^1$, i.e. of the central charge associated with that particular boundary component). As it turns out, this minimal assumption is sufficiently strong to ensure rigidity of the following two characteristics of the theory:

\begin{itemize}
\item[1)] the possible values allowed for $\lambda$, for given topological data (connectivity, as controlled by $n \ge 2$) and geometric data: specifically, we obtain a strictly positive lower bound, 
$\lambda \ge \lambda_{m} = 2A/P > 0$;
\item[2)] the complex geometry associated with the domain $\Omega$ (i.e. the intrinsic 2D metric of $\Omega$ and its boundary components $\{ \Gamma_k\}$), which becomes completely determined by imposing that the lower bound is achieved; moreover, this identifies the pseudo-Riemannian trajectories on $\overline{\Omega}$ to the family of positive trajectories of a quadratic differential induced by the requirement $\lambda = \lambda_{m}$.
\end{itemize}
 
The relevance of the results for low-dimensional quantum theory stems from revealing the connections  between topology (specifically, homological data), complex geometric structure (selection of a specific element from the Teichm\"{u}ller space of the domain), and finally the parameter of deformation quantization, $\lambda$, which ultimately enters all aspects of the associated quantum field theory. Another non-trivial aspect lies in the introduction of a geometric characteristic (``length scale") $\lambda_m$, controlled by the geometric data of $\overline{\Omega}$, compatible with the boundary quantization of each component of $\p \Omega$, which for a trivial topological structure of $\Omega$ (i.e., simply-connected) would normally correspond to the 1-cocycle of projective representations of Diff$(S^1)$, whose values may be any real number (no ``length scale"). 
 
 This paper is structured as follows: the main results are presented in detail in Section 2, where we also briefly review the general theory behind deformation quantization, structure of the Teichm\"{u}ller space and choice of complex structure of the domain, and the relationships between pseudo-Riemannian metric and quadratic differentials. Section 3 contains the proof of the main theorem which establishes that the lower bound of the quantization parameter is achieved only for domains with connectivity $n = 2$, and that the domain can only be an annulus in this case. Remarks regarding the applicability of these results to open problems in low-dimensional quantum theories are provided in Section 4; the remaining proofs are collected in Appendix 1. 
 
\section{Main results} \la{main}

\begin{figure}[h!!] \begin{center}
\includegraphics*[width=11cm]{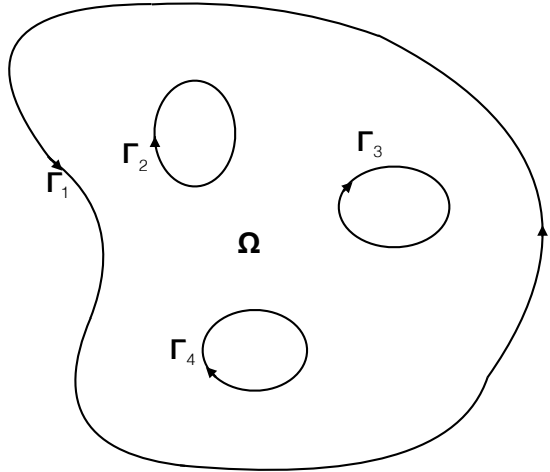}
\caption{The domain $\Omega$ and its boundary components, shown with their orientations relative to $\Omega$ (clockwise for the interior contours, counterclockwise for the exterior one).}
\label{fig1}
\end{center}
\end{figure}

Let $\Omega \subset \mathbb{C}$ be a bounded domain, with boundary components $\Gamma_k, k = 1, 2, \ldots, n$, $n \ge 2$, and denote by $A$ the area of $\Omega$ and by $L_k$ the total length of the curve $\Gamma_k$ (in the more general case we replace the area and arclength measures of $\Omega$ and $\p \Omega$ by their respective Minkowski contents), as illustrated in Figure~\ref{fig1}.  

If $\tau_k(z)$ represents the unit tangent vector at $z \in \Gamma_k$, according to the canonical counter-clockwise orientation on $\Gamma_k$,  denote by $\{ \Omega_k \}_{k=1}^n$ the domains defined by
\be
\Omega_k \cap \Omega = \emptyset,  \quad \partial \Omega_k = \Gamma_k,
\ee
and choose $\Omega_1$ for the one which is unbounded.  Therefore, along $\Gamma_1$ the inward-oriented (with respect to $\Omega$) unit normal vector is given by
\be
n_1(z) = i \tau_1(z),
\ee
while on all the ``interior" contours $\Gamma_k (k = 2, 3, \ldots, n)$ we have the usual defining relation for the inward-oriented (with respect to $\Omega$) unit normal:
\be
n_k(z) =- i \tau_k(z), \quad k = 2, 3, \ldots, n.
\ee

A convenient compact notation is obtained by introducing the parameter $\alpha$ as follows:
$$
\alpha = 1 \mbox{ for } k \ge 2, \quad \alpha = -1 \mbox{ for } k = 1,
$$
so that $n_k(z) =- i \alpha \tau_k(z), k = 1, 2, \ldots, n$. Let $A(\Omega)$ denote  the algebra of analytic functions on $\Omega$.

\subsection{Quasiconformal transformations and deformation quantization for a single boundary curve}

In the case of a single boundary component, diffeomorphic to the unit circle $S^1$, canonical geometric quantization is equivalent \cite{Zograf, Iwaniec} to the unitary representations of the central extension of the group ${\rm{Diff}} (S^1)$ by the group of real numbers $\mathbb R$ (the Virasoro-Bott group), where ${\rm{Diff}} (S^1)$ represents the group of orientation preserving $C^{\infty}$ diffeomorphisms of the unit circle $S^1$, with the group operation given by composition of diffeomorphisms. Denoting by  
${\rm{Vect }} (S^1)=\{\phi=\phi(\theta)\frac{d}{d\theta}\mid \phi\in C^{\infty}(S^1\to\mathbb R)\}$ the space of smooth real vector fields on the circle, equipped with the Lie brackets $[\phi_1(\theta)\frac{d}{d\theta},\phi_2(\theta)\frac{d}{d\theta}]$, this leads in turn to the non-trivial extension of ${\rm{Vect} } (S^1)$ by the algebra of real numbers. The central extension is unique up to isomorphisms and is given by the Gelfand-Fuchs cocycle
\be
\omega(\phi_1,\phi_2)=\oint_{S^1}\phi_1'(\theta)\phi_2''(\theta)d\theta.
\ee

Identifying the unit disk $\mathbb{D}$ with the upper hemisphere of the Riemann sphere, the universal Teichm\"{u}ller space is the  space of real analytic quasiconformal mappings of $\mathbb{D}$ into itself \cite{Iwaniec}, defined up to a M\"obius transformation for the boundary restrictions. Its elements form an open subspace of the bounded holomorphic functions on $\mathbb{D}$ with the $C^{\infty}$ norm. For a general compact Riemann surface $\Sigma$ of finite genus, the Teichm\"{u}ller space  is the  subspace of the universal Teichmüller space invariant under the fundamental group $\pi(\Sigma)$.  It is known \cite{Bers} that the  Teichm\"{u}ller space  of $\Sigma$  is equivalent to is an open subspace of the complex vector space of quadratic differentials on $\Sigma$, invariant under $\pi(\Sigma)$. Owing to the tensor covariance of quadratic differentials,  this means that the adjoint action on the dual of ${\rm{Diff}} (S^1)$ contains operators of the form 
$$
\mathcal{L}_f \equiv \frac{d^2 \,\,}{d \theta^2} + \frac{1}{2}\mathcal{S}(f)(\theta),  
$$
where $f \in {\rm{Diff}} (S^1)$ and $\mathcal{S}$ represents the Schwarzian derivative. The invariance of the Schwarzian under composition by M\"obius  transformations (whose Schwarzian derivative is identically zero) provides a bijective map between the operators $\mathcal{L}_f$ and the elements of the  Teichm\"{u}ller space, such the null element corresponds to $\mathcal{L}_0$ (so it can be interpreted as  the quantization of the simple harmonic oscillator). We will be using this correspondence in the main derivations of this paper.

As shown in \cite{Radulescu}, the (Murray - von Neumann) algebraic deformation quantization of the space $\mathbb{H}/\pi(\Sigma)$ (or equivalently $\mathbb{D}/\pi(\Sigma)$, with the same hyperbolic geometry), also known as Berezin quantization,  leads to a continuous family of von Neumann algebras, indexed by a real parameter $t \in \mathbb{R}_+$ (quantization parameter), which corresponds to the unrestricted real central extension of the theory in the geometric deformation approach indicated earlier. The algebraic deformation construction  also produces a  Hochschild 2-cocyle \cite{Radulescu}, to be compared to the  Gelfand-Fuchs cocycle mentioned above.

Therefore, in the case of quasiconformal transformations of a Riemann surface of finite genus, defined up to M\"obius transformations on a single boundary component, deformation quantization has no intrinsic ``length scale", and no associated ``finite gap", a feature which disappears when considering multiple boundary components. 

\subsection{Geometric quantization for bounded domains with multiple boundary components}

To formulate the problem in the simplest terms, we are seeking a set of $n$ vector-valued functions (vacuum states) $\Psi_k(z) = (v_k(z), w_k(z)):\overline{\Omega} \to \mathbb{C}^2, k = 1, 2, \ldots, n$, and an analytic function $\varphi \in A(\Omega)$, such that the following conditions are satisfied: 

\begin{itemize}
\item[1)] (Analyticity) The connection form $A_{\varphi} dz$ and the corresponding covariant derivative 
$$
\nabla_{\varphi} = \mathbb{I} \frac{d \,\,}{dz} + \frac{1}{\lambda}A_{\varphi},  
$$
(where $\mathbb{I}$ is the $2 \times 2$ unit matrix), with  
\be \la{1}
A_{\varphi}(z) \equiv
\left (
\begin{array}{cr}
 0 & -1 \\
{\varphi'(z)} & 0 \\
\end{array}
\right ),
\ee
are well-defined on $\Omega$; 
\item[2)]  (Vacuum vectors) For every boundary component $\Gamma_k, k = 1, 2, \ldots, n$:
\be \la{2}
\nabla_{\varphi}  \Psi_k(z) = 0, \quad z \in \Omega.
\ee
\item[3)] (Boundary symplectic forms) There exist linear transformations \\ $U_k(z): \overline{\Omega} \to GL(2, \mathbb{C})$, whose boundary values are unitary, \\ $U_k(z)\Big |_{\Gamma_k} \equiv U_k(s) \in U(2)$, such that, on the sections $\Psi_k(z(s))$: 
\be \la{3}
\left [ \mathbb{I} \frac{d \,\,}{ds} + \frac{1}{\lambda} 
\left (
\begin{array}{rr}
 0 & -1 \\
1 & 0 \\
\end{array}
\right )
\right ] (U_k \Psi_k)(s) = 0,
\ee
where $s$ denotes the local arclength parameter on $\Gamma_k$.
\end{itemize}


\subsection{Discussion of the choice of gauge field and boundary conditions}

\subsubsection{The analyticity condition}

As mentioned already, this minimal model is chosen such that for $n = 1$ it can be reduced to the known problem of Teichm\"uller space for a Riemann surface of genus 0 \cite{Zograf}, for which the boundary $\Gamma$ is the intersection of the upper and lower hemispheres (and can be chosen $\Gamma = \mathbb{T} = S^1$ without loss of generality). In that sense, the present model calls for a simultaneous quantization on $n \ge 2$ boundaries; it can be regarded as the generalization corresponding to an algebraic variety (with $n$ connected components) instead of a Riemann surface. 

In the simplest implementation, the single-boundary model would correspond to the simple harmonic oscillator, for which the covariant derivative will take the form given in  \eqref{3} and the vacuum representation $U(1) \oplus U(1)$ is given by 
\be \la{local}
v(s) = e^{isE_0/\lambda},\quad w(s) = i v(s), \quad s \in [0, 2\pi),
\ee
and $\lambda = E_0 = 1$, setting the vacuum energy to $E_0=1$, as we indicate in more detail later in this section. Therefore, the connection form reduces to the symplectic form in two dimensions \eqref{3}, a condition imposed in the general case as well.  
 

The choice of connection form $A_{\varphi}$ reflects the requirement that the quantization on any boundary component  $\Gamma_k$ admits an analytic continuation in $\Omega$ to any other  boundary component, and moreover that the condition 
\be \la{problem}
v_k'' + \frac{\varphi'}{\lambda^2} v_k = 0,
\ee
is the same for all $v_k, k = 1, 2, \ldots, n.$ Then let us also note that $\varphi'$ is proportional to the Schwarzian of the ratio of any pair  $v_{i},v_{j}, i \ne j$ \cite[Ch. 6]{LG}:
$$
\varphi' = \frac{\lambda^2}{2} \mathcal{S}\left ( \frac{v_i}{v_j}
\right ) ,
$$
where
$$
\mathcal{S}(f) \equiv \left (\log f' \right )'' - \frac{1}{2}[(\log f')']^2.
$$
As a Schwarzian, $\varphi ' dz^2$ is a quadratic differential \cite{QD2}--\cite{QD4}, which means that the solution to the problem \eqref{problem} provides at once the compatible quantization for each boundary component $\Gamma_k$, and a unique element of the universal Teichm\"{u}ler space, corresponding to the quadratic differential $\varphi'(z)dz^2$. In other words, solving for the unique quantization on multiple boundaries leads to a specific complex structure on $\Omega$, corresponding to the quadratic differential of density $\varphi'(z)$. As we will see in the next section, this complex structure induces (pseudo-)Riemannian metrics on all boundary components $\Gamma_k$. For now, we only make the remark that, for $\varphi(z) \in A(\Omega)$, the 2-form 
$
d\varphi \wedge dz = \overline{\p} \varphi d \bar z \wedge dz 
$ 
vanishes identically, while the quadratic differential is the corresponding (non-vanishing) symmetric tensor product. Relaxing the conditions on $\varphi$ and allowing for simple poles $z_j \in \Omega$ of real residues $c_j \in \mathbb{R}$, the 2-form 
$
d\varphi \wedge dz =  2\pi \sum_{j} c_j \delta(z-z_j) dx\wedge dy
$
would correspond to a pseudo-Riemannian metric with the curvature concentrated at the poles $z = z_j$. 
Finally, regarding the holomorphic covariant derivative $\nabla_{\varphi}$, a simple calculation yields for its curvature form the expression 
$$
\mathcal{R} = [\nabla_{\varphi}, \overline{\nabla}_{\bar \varphi}] dz \wedge d \bar z = 
\Im \left (
\begin{array}{rr}
 \frac{ \varphi'}{\lambda^2} & 0 \\
 \frac{1}{4\lambda} \Delta \varphi & - \frac{\varphi'}{\lambda^2} \\
\end{array}
\right ) dx \wedge dy, 
$$
where $\Delta$ denotes the usual Laplace operator in $\mathbb{R}^2$. Therefore, in the case of analytic connections $\varphi \in A(\Omega)$ (but also for the cases where $\varphi$ is allowed to have isolated logarithmic and pole singularities as indicated above) $\Delta \Im \phi = 0$, so there is only one nontrivial term controlling the connection curvature, namely $\varphi '(z)$. 

\subsubsection{Vacuum states} As stated, the vacuum states (for each boundary component) are chosen in the lowest representation corresponding to a connection proportional to the symplectic form in two dimensions \eqref{2}. Moreover, they are annihilated by the derivative $\nabla_{\varphi}$, and are related to the local canonical vacua on the respective boundary  \eqref{local} by unitary transformations $U_k \in U(2)$. While not explicitly stated amongst the conditions \eqref{1}-\eqref{3}, the vacua $\Psi_k$ are cyclic vectors in their respective Hilbert spaces of states, and therefore cannot vanish anywhere on $\Gamma_k$. 

We recall \cite{Connes} that, in order to express explicitly the cyclic property of vacuum states,  it is necessary to know the specific functional space they generate, i.e. the Hilbert space of states $H$ and the norm used on $H, || \,.\,  ||_{H}$; some known examples are provided in Section 4. Here we only remark that, while the full might of functional models is indeed needed to compute estimates for all the observables of the theory, it is fortuitously not required for the purpose of this investigation.  

\subsubsection{Reduction to canonical symplectic form and monodromy} 

The justification of condition \eqref{3} follows directly from the simplest model at $n=1$, i.e. the simple harmonic oscillator. In normalized dimensionless form, the classical limit of the theory corresponds to the phase-space representation $(p, q) \in S^1$, or Hamilton equations written in the form \eqref{3}, with arclength parameter $s$ playing the role of the time coordinate, for periodic trajectories with period $2\pi$. 

The $U(1)$ representation of the geometric quantization is equivalent with the monodromy condition exp$[\frac{i}{\lambda}\oint p dq] = 1$, or else $E_n - E_0 = n \lambda$, where $E_0$ is the ground-state (vacuum vector) energy. For the $U(1) \oplus U(1)$ representation, the monodromy condition can be applied to the vacuum state directly, since the ground-state energy then becomes $E_0 = 1$ in normalized units. We note that global monodromy for $n \ge 2$ is much weaker than for $n = 1$ and therefore insufficient to provide the solution. The generalized monodromy condition for $n \ge 2$ is:
\be \la{mon}
\sum_{k=1}^n \oint_{\Gamma_k} d Arg(v_k) = 0 \mod  2\pi  \mathbb{Z}. 
\ee

\subsection{Summary of the main results}

The two main results of this analysis are:

\begin{thm}[Topological bound] \la{th1}
For any quantization of a bounded domain with $n \ge 2$ boundary components, area $A$ and perimeter $P$, the quantization parameter 
$$
\lambda  \ge  \lambda_m  = 2A/P = \inf_{\varphi \in A(\Omega)} ||\bar z - \varphi(z) ||,
$$ 
where the norm on the RHS is the uniform norm for functions continuous on $\Omega$.  
\end{thm}


\begin{thm}[Topological restriction] \la{th2}
Under the assumptions of Theorem~\ref{th1}, and additional conditions that each boundary component $\Gamma_k$ is a real-analytic simple curve, if the topological bound for the quantization parameter $\lambda_m$ is achieved, then:
\begin{itemize}
\item[1)] $n = 2$ and the corresponding domain $\Omega_m$ is an annulus 
$$
\Omega_m = \{ z \in \mathbb{C} | \,\, R_1 > |z| > R_2\}.
$$
\item[2)] The minimal value $\lambda_m = R_1-R_2$. 
\end{itemize}
\end{thm}

In the next section we discuss in detail the proof and implications of the topological restriction Theorem~\ref{th2} above. Following that discussion, we return to the first result and discuss the foundation of the general lower bound given by Theorem~\ref{th1}.

\section{The theory at lower bound of deformation parameter: unique selection of complex structure and maximally-symmetric domain}

\subsection{Notations and defining equations} \la{intr}

Let $\Omega$ be the domain, with boundary components $\Gamma_k, k = 1, 2, \ldots, n$ and corresponding Schwarz functions $S_k(z)$, so that $S_k(z) = \bar z$ on $\Gamma_k$ and
\be
\sqrt{S'_k(z)} = \frac{ds}{dz} = u_k(z) = \frac{1}{\tau_k(z)}, \quad z \in \Gamma_k,
\ee
where the unit tangent vectors are defined using the convention established in \S \ref{main}. 

Conditions \eqref{1} -\eqref{3} lead by direct computations to the following Riccati equations for the functions $u_k(z)$:
\be \la{riccati}
u_k^2 + i\alpha \lambda u_k' = \varphi'(z), \, z \in \Omega.
\ee
It is known \cite{KB} that this problem is equivalent to 
\be \la{start}
S_k(z) + \lambda \bar n_k(z) = \varphi,
\ee
or equivalently
\be  \la{start1}
S_k(z) + i \lambda u_k(z) = \varphi, \quad k = 2, \ldots, n,
\ee
and
\be \la{start2}
S_1(z) - i \lambda u_1(z) = \varphi,
\ee
where $\lambda = 2A/P$, with $A$ the area of $\Omega$ and $P$ its perimeter, $P = \sum_k L_k$, and $L_k$ the perimeter of $\Gamma_k$.  Note that we can also write Eqs.~\eqref{start}-\eqref{start2} in the form
\be \la{gm}
\bar z + i\alpha \lambda \dot {\bar{z}}(s) = \varphi, \quad \dot{f} \equiv \frac{df}{ds},  \quad \alpha = 1 \mbox{ for } k \ge 2, \quad \alpha = -1 \mbox{ for } k = 1.
\ee
Moreover, the parameter $\lambda$ happens to be the lower bound for the distance from $\bar z$ to the space of analytic functions on $\Omega$, in the uniform norm \cite{KB}.

\subsection{Quadratic differentials and classification of extremal domains}


\begin{thm} \la{function}
In $\Omega$, $\varphi'(z)dz^2$ is a quadratic differential, real-valued on $\p \Omega$, and
\be
\varphi'(z) dz^2 = (1+ \alpha \lambda \kappa) ds^2
\ee
along each component $\Gamma_k$ of $\p \Omega$ (where $\alpha$ is defined as in Eq.~\eqref{gm}). 
Moreover, on every component $\Gamma_k$ of $\p \Omega$, 
$
\oint_{\Gamma_k} (1 + \alpha \lambda \kappa) ds > 0.
$
\end{thm}
\begin{proof}
From Eq.~\eqref{riccati} and $\alpha^2=1$, we notice that the functions
\be \la{linsol}
v_k(z) \equiv \exp \left [-\frac{i \alpha}{\lambda}  \int^z u_k(\zeta) d\zeta \right ], \quad k = 1, 2, \ldots, n,
\ee
solve the linear second-order differential equation associated to Eq.~\eqref{riccati}
\be \la{ODE}
v'' = -\frac{\varphi'}{\lambda^2} v,
\ee
so $\varphi'$ is proportional to the Schwarzian of the ratio of any pair  $v_{i},v_{j}$ \cite[Ch. 6]{LG}:
$$
\varphi' = \frac{\lambda^2}{2} \mathcal{S}\left ( \frac{v_i}{v_j}
\right ) ,
$$
where
$$
\mathcal{S}(f) \equiv \left (\log f' \right )'' - \frac{1}{2}[(\log f')']^2.
$$
As a Schwarzian, $\varphi ' dz^2$ is a quadratic differential \cite{QD2}-\cite{QD4}.
From Eq.~\eqref{start}
we obtain by differentiation with respect to arclength:
\be
\bar \tau(s) + \lambda \alpha \kappa \bar \tau(z(s)) = \varphi'(z)\cdot \tau(z(s)),
\ee
where (following the convention introduced in \S~\ref{intr}),
$$
\kappa \equiv -i \bar \tau \cdot \frac{d\tau}{ds}
$$
is the curvature of $\Gamma_k$ at that point. Multiplying both sides by $\tau(s)$, we get
\be \la{geom}
1 + \lambda \alpha \kappa = \varphi'(z) \tau^2(z),
\ee
or equivalently
\be \la{l0}
{\varphi'(z)}dz^2 = (1+\lambda \alpha \kappa) ds^2,  \quad z \in \Gamma_k.
\ee
%
%
%

For any interior contour $\Gamma_{k \ge 2}, \alpha = 1$, so we evaluate
\be
\oint_{\Gamma_k} (1+ \lambda \kappa)ds = L_k + 2\pi\lambda > 0,
\ee
hence the quadratic differential is strictly positive-definite on $\Gamma_k$. Since $\alpha = -1$ for $\Gamma_1$, compute
\be
\oint_{\Gamma_1} (1- \lambda \kappa)ds = L_1 - 2\pi\lambda = L_1 - \frac{4\pi A}{P},
\ee
with $A = \mbox{Area}(\Omega)$ and $P = L_1 + \sum_{k \ge 2}L_k$ its perimeter. Using $P \ge L_1$,
\be
L_1 - \frac{4\pi A}{P} \ge L_1 - \frac{4\pi\mbox{Area}(\Omega)}{L_1} \ge \frac{4\pi}{L_1}\left [ \mbox{Area}(\Omega_1^c) -
\mbox{Area}(\Omega) \right ] \ge 0,
\ee
where we have used the isoperimetric inequality for the complement of $\Omega_1$, $\Omega_1^c$, and the fact that
$\Omega \subseteq \Omega_1^c$. 
\end{proof}

\begin{definition} \la{nw1}
Let $\Sigma^{\pm}$ be the union of Stokes and anti-Stokes graphs of Eq.~\ref{ODE} in $\Omega$ \cite[Lemma~9.2-1]{Olver}, i.e. the union of arcs  $\{ \gamma^{\pm}_j\}$ satisfying
$$
\Im \int_{z_0}^z \sqrt{\varphi'(\zeta)} d\zeta = 0, \quad \zeta \in \gamma^{+}_j \subset \Sigma^{+}, \quad 
\Re \int_{z_0}^z \sqrt{\varphi'(\zeta)} d\zeta = 0, \quad \zeta \in \gamma^{-}_j \subset \Sigma^{-},
$$
where $z_0$ is any zero of $\varphi'(z)$ in $\overline{\Omega}$. 
\end{definition}

It is known \cite{Olver}-\cite{D} that if $\varphi'(z)$ is analytic in $\Omega$, then $\Sigma^{+}, \Sigma^{-}$ have the same number of arcs $\gamma^{\pm}_j$, they intersect only at zeros of $\varphi'(z)$, and each arc $\gamma^{\pm}_j$ is analytic, with one endpoint being a zero of $\varphi'$, and the other being either another zero, or a point on $\p \Omega$ (or possibly, both). Moreover, at a zero $z_0 \in \Omega$ of $\varphi'$ or order $m \ge 1$, there are exactly $m+2$ arcs from $\Sigma^{+}$ with local angle between adjacent arcs equal to $2\pi/(m+2)$, and another $m+2$ arcs from $\Sigma^{-}$, each of them bisecting the angle between two consecutive arcs of $\Sigma^{+}$.  


Let $z_0 \in \overline{\Omega}$ be a zero of order $m$ of $\varphi'$. By elementary calculations, it is easy to show that the general solution of Eq.~\eqref{linsol} has the local power series expansion
\be \la{entire}
v(z) = C_1 + C_2(z-z_0)+
C_3 (z-z_0)^{m+2}+ O((z-z_0)^{m+3})
\ee
with the coefficients determined by the local power expansion of $\varphi'(z)$, and $C_{1, 2, 3}$ all non-vanishing if $z_0 \in \p \Omega$. In the case where $\varphi'$ is a monomial, the solution \eqref{entire} is a linear combination of Bessel functions of orders $\pm \frac{1}{m+2}$, and more generally if $\varphi'$ is a polynomial, the solution is an entire function \cite[Ch. 7]{Olver}. However, the local solution is not convenient to use when exploring global properties of solutions such as $|v_k(z)|_{\Gamma_k} = $ constant, satisfied by Eq.~\eqref{linsol}. Instead, we are lead to the asymptotic series representations, valid outside a small neighborhood of $z_0$. We briefly review here the general theory. 

Defining the local coordinates $\zeta \equiv \epsilon(z-z_0)$, with $\epsilon$ a scale parameter, arbitrarily small but strictly positive, then c.f. \cite[Ch. 6]{LG}, \cite[Ch. 3]{Olver}, \cite{FF}, the general solution for Eq.~\eqref{ODE} admits the asymptotic series representation known as Liouville-Green (LG) or Jeffreys-Wentzell-Kramers-Brillouin (JWKB):
\be \la{stokes}
v(\zeta,\epsilon) = \frac{\sqrt{\lambda}}{(\varphi')^{1/4}} \left [ C_1 e^{\frac{i}{\lambda \epsilon}\int_0^\zeta  \sqrt{\varphi'} d\xi} + C_2 e^{-\frac{i}{\lambda \epsilon}\int_0^\zeta  \sqrt{\varphi'} d\xi} \right ][1 + O(\epsilon)],  
\ee
where $C_{1, 2}$ are constants, and $\zeta$ belongs to a domain $D$ having 0 as boundary point. In particular, for $z \in \Sigma^{+}$, the domain of validity includes a wedge domain of angle $2\pi/(m+2)$, with $\Sigma^{+}$ bisecting the angle.  The solution is approximated by the asymptotic expansion in the sense of the Borel-Ritt theorem \cite[Ch. 3]{Olver}, i.e. the R.H.S. of Eq.~\eqref{stokes} is an entire function of $\zeta$, smooth in both  $\zeta$ and $\epsilon$, and 
\be \la{conv}
\lim_{\epsilon \to 0}\frac{1}{\epsilon} \left [ v(\zeta, \epsilon) - \frac{\sqrt{\lambda}}{(\varphi')^{1/4}} \left (C_1 e^{\frac{i}{\lambda \epsilon}\int_0^\zeta  \sqrt{\varphi'} d\xi} + C_2 e^{-\frac{i}{\lambda \epsilon}\int_0^\zeta  \sqrt{\varphi'} d\xi} \right ) \right ] = 0, \, \zeta \in D.
\ee
The exact determination of the coefficients of the asymptotic series depends on the required asymptotics of the solution at $|\zeta| \to \infty$, and the precise choice which provides a single-valued solution with given asymptotic behavior must take into account the delicate balance of the first two terms in the R.H.S. of Eq.~\eqref{stokes}, especially when crossing any arc of $\Sigma^{+}$ \cite{D,FF}, where the dominant and sub-dominant series may change character (``Stokes phenomena").

\begin{lemma} \la{n2}
The function $\varphi'$ cannot vanish at any point on $\p \Omega$, so the quadratic differential $\varphi'(z) dz^2$ is strictly positive-definite on $\p \Omega$ (i.e. it is a Riemannian metric). 
\end{lemma}
\begin{proof}
Assume that $\varphi'(z_0) = 0, z_0 \in \Gamma_k \in \p \Omega$. Then from Definition~\ref{nw1} and Theorem~\ref{function}, $\Gamma_k \subset \Sigma^{+}\cup \Sigma^{-}$. The two arcs $\gamma_{1,2}(z_0)$ of $\Gamma_k$ meeting at $z_0$ are elements either of $\Sigma^{+}$ or of $\Sigma^{-}$. However, at least  one such arc must belong to $\Sigma^{+}$, because otherwise $\Gamma_k \subset \Sigma^{-}$, which implies that $\varphi' dz^2$ is negative-definite on $\Gamma_k$, so according to Eq.~\eqref{l0} $1 + \alpha \lambda \kappa \le 0$ everywhere on $\Gamma_k$, which contradicts Theorem~\ref{function}.

Take now $z$ on the arc belonging to $\Sigma^{+} \cap \Gamma_k$. We may assume that the arc has finite length, as  Remark~\ref{arc} shows.  According to the LG formula \eqref{stokes}, the solution \eqref{linsol} has the asymptotic expansion 
\be
v_k(z_0 + \epsilon \zeta) = 
\frac{\sqrt{\lambda}}{(\varphi')^{1/4}} \left [ C_1 e^{\frac{i}{\lambda \epsilon}\int_0^\zeta  \sqrt{\varphi'} d\xi} + C_2 e^{-\frac{i}{\lambda \epsilon}\int_0^\zeta  \sqrt{\varphi'} d\xi} \right ][1 + O(\epsilon)]
\ee
with $C_{1,2}$ constants. Denote by $\gamma = \Sigma^{+} \cap \Gamma_k \cap D$, and notice that along $\gamma$, condition \eqref{conv} and Eq.~\eqref{linsol} give
\be \la{kn}
\lim_{\epsilon \to 0}\frac{1}{\epsilon} \left [ e^{-i\frac{\alpha}{\lambda \epsilon}s(\zeta)} - \frac{\sqrt{\lambda}}{(\varphi')^{1/4}} \left (C_1 e^{\frac{i}{\lambda \epsilon}\int_0^\zeta  \sqrt{\varphi'} d\xi} + C_2 e^{-\frac{i}{\lambda \epsilon}\int_0^\zeta  \sqrt{\varphi'} d\xi} \right ) \right ] = 0.
\ee
Take $z \in \gamma$ so that the arclength along $\gamma$ from $z_0$ to $z$, is $s > 0$. Let $\omega(s) \equiv \int_{z_0}^z  \sqrt{\varphi'} d\xi$ and note that $\omega(s) > 0$ from Thm.~\ref{function}.  Also, let $K_{1, 2} \equiv 
\frac{\sqrt{\lambda}}{(\varphi'(z))^{1/4}} C_{1,2}$ and consider first the case of an interior boundary component $\Gamma_k$, i.e. $\alpha = 1$. Condition \eqref{kn} implies then 
\be \la{yo}
\lim_{\epsilon \to 0}\left |1  - K_1 e^{\frac{i}{\lambda \epsilon}(s+\omega(s))} - K_2 e^{\frac{i}{\lambda \epsilon}(s-\omega(s))}\right | = 0.
\ee
Taking now the sequence $\epsilon_n \equiv \frac{s + \omega(s)}{2 \pi \lambda n}, n \in \mathbb{N}$, we obtain 
\be
\lim_{n \to 0}\left |1  - K_1 - K_2 q^n\right | = 0, \quad 
q = e^{2\pi i\frac{s-\omega(s)}{s + \omega(s)}} \in \mathbb{T}.
\ee
This is possible either if $K_1 = 1, K_2 = 0$ for arbitrary $q$, or if $K_1 = 0, K_2 = 1, q=1$. 

Since the point $z \in \gamma$ was arbitrary,  $K_1+K_2 = 1$, so $|\varphi'(z)| = |1 + \alpha \lambda \kappa| = $ const. along the arc $\gamma$. But then $\Gamma_k$ contains a circular arc $\gamma$. It is known \cite{KB} that in this case $\Omega$ is either a disc or an annulus and $\varphi'(z) \ne 0$ on $\p \Omega$, a contradiction. 

For the case of the exterior boundary $\alpha = -1$, we exhange $K_1$ and $K_2$ in Eq.~\eqref{yo} and the argument follows identically. 
\end{proof}

\begin{remark} \la{arc}
If the arc of boundary (and of the Stokes graph) connecting two consecutive zeros of $\varphi'(z)$ does not have finite length, then by analyticity of $\Sigma^{+}$ and $\p \Omega$, there is an arc of boundary where $\varphi'(z)$ vanishes identically, and therefore is a circular arc, which is a contradiction as explained in Lemma~\ref{n2}. 
\end{remark}

\begin{thm}
The domain $\Omega$ is a {\emph{maximal domain}} in the sense of \cite{Jenkins}, so its connectivity (and the total number of boundary components of $\p \Omega$) is 1 or 2.
\end{thm}
\begin{proof}
From Lemma~\ref{n2} and the general properties of $\Sigma^{+}$ listed above, it follows that there are no open arcs of $\Sigma^{+}$ in $\overline{\Omega}$. Otherwise, there would be arcs ending on $\p \Omega$ (not allowed by Lemma~\ref{n2}), or vertices of degree 1 in $\Sigma^{+}$ (not allowed since they would correspond to poles of $\varphi'$, which was assumed to be analytic in $\Omega$). Therefore, any trajectory (in the sense of \cite{Jenkins}) of the quadratic differential $\varphi'(z) dz^2$ (including $\Sigma^{+}$) 
can be extended to a closed curve in $\p \Omega$, so $\Omega$ is a {\emph{maximal domain}}. Then from \cite[Theorem 1]{Jenkins}, the connectivity of $\Omega$ cannot exceed 2. 
\end{proof}

Theorem~\ref{th2} then follows for the case $n=2$; since its proof is rather technical and requires a number of additional results, it is presented in full in Appendix 1, under the equivalent formulation Theorem~\ref{final}. Here we use the explicit solution to illustrate the global monodromy condition \eqref{mon}: 
$$
v_1(z) = \left ( \frac{z}{R_1} \right )^{\frac{R_1}{R_1-R_2}}, \quad 
v_2(z) = \left ( \frac{z}{R_2} \right )^{-\frac{R_2}{R_1-R_2}}, \quad 
\varphi(z) = \frac{1}{\pi}\cdot \frac{A}{z}, 
$$
so that $v_1(z)$ can be extended away from $\Omega$ for $z \to 0$, and $v_2(z)$ can be extended away from $\Omega$ for $z \to \infty$, and \eqref{mon} is satisfied:
$$
\sum_{k=1}^n \oint_{\Gamma_k} d Arg(v_k) = 2 \pi. 
$$
\section{Deformation quantization for the generic case}

Under the conditions considered in this work, the (von Neumann) deformation quantization is equivalent to finding a Hilbert space $H$ and a $\star$- algebra of bounded operators on $H$ (a subset of the algebra of bounded operators $B(H) \subset L(H)$, where $L(H)$ denotes the space of linear operators on $H$), closed in the weak operator topology, containing the identity, and equipped with a well-defined trace functional. For the resulting von Neumann algebra to be non-commutative, the space $H$ must be at least 2-dimensional (as considered throughout this paper). 

Endowed with this structure, the operator (sometimes called Bergman operator) of multiplication by the holomorphic coordinate $z$, $T_z \in L(H)$, and its adjoint $T^{\dagger}_z$, have a non-zero commutator $[T_z, T^{\dagger}_z] = D$, where $D$ is a positive, trace-class operator (i.e. the operator $T_z$ is hyponormal). In the simplest case, the Heisenberg canonical commutation relations are imposed so that on a subspace of $H$ (where the restriction of $T_z$ is purely hyponormal), 
$$
[T^{\dagger}_z, T_z] = \lambda \mathbb{I}.
$$

Taking for an obvious example the simple harmonic oscillator, i.e. the case $\Omega = \mathbb{D}, \Gamma = S^1$, a simple calculation leads to the geometric quantization solution 
$$
\lambda = 1, \quad \varphi(z) = 0, \quad \Psi(z) = (z, 1), \quad U(s) = {\rm{diag}} (1, \tau(s)). 
$$ 
Clearly, the adjoint operator to $T_z^{\dagger} = \nabla_{\varphi = 0}$ is given by $T_z = z \mathbb{I}$, and the cyclicity of the vacuum $\Psi(z)$ leads to the Hilbert space which is the closure, in $L^2(S^1)$ norm, of the monomials $T_z^{n}\Psi(z)$, i.e. the orientation-preserving modes $e^{in\phi}, \phi \in S^1, n \in \mathbb{N}$. 

\subsubsection{Bargmann-Fock space}

In special situations (see, e.g. the review in \cite{Fedosov}), deformation quantization reduces to the case known as Weyl quantization, where $H$ is an $L^2$-space. A familiar realization corresponds to  the case where $\Omega = \mathbb{C}$, and the Hilbert space is the Segal-Bargmann space (or Bargmann-Fock space) of holomorphic functions with the  inner product 
$$
\langle f | g \rangle \equiv \frac{1}{\pi}\int_{\mathbb{C}} \overline{f(z)} g(z) e^{-\frac{|z|^2}{\lambda}} dxdy,
$$
and corresponding norm. In this case, the canonical commutation relations correspond to the fact that $T^{\dagger}_z = -\lambda \frac{\p }{\p z}$ (where the integrals are defined in distributional sense). The quantization parameter $\lambda$ is ``free", in the sense that any positive value $\lambda > 0$ is allowed by the theory, and moreover the limit $\lambda \to 0$ is well-defined, and leads to the expected ``classical" theory, in which quantum states normally described by smooth distributions (with exponential decay at $\infty$) degenerate into Dirac point masses, i.e. ``classical" point particles. In this case, $\lambda$ is simply a scale parameter characterizing the ``spread" of wave functions (or rather of their associated density functions), and can be taken to be an arbitrary number $\lambda \in \mathbb{R}_+$.  

\subsubsection{Generalized Weyl quantization}

In the generic case of bounded domain with multiple boundary components,  the convenient simplicity of the Bargmann space is not available any longer, and instead determination of the adjoint operator $T^{\dagger}_z$ requires knowledge of  complicated boundary terms, depending on the specific space of distributions used in the construction. The following example illustrates this. 

Let $H$ be a Hilbert space (possibly finite-dimensional, e.g. $\ell^2(\mathbb{C})$). We introduce the following functional spaces (after \cite{MP}): 
\begin{definition} \la{de}
With $dxdy$ representing the usual  Lebesgue area measure, 
\begin{itemize}
\item[(i)] 
$$
L^2(\Omega, H) \equiv \left \{f: \Omega \to H \Big | \int_{\Omega} f dxdy < \infty \right \},
$$  
\item[(ii)] 
$$
W^m_{\bar \p} (\Omega, H) \equiv  \left \{f \in L^2(\Omega, H) | (\bar \p)^k f \in L^2(\Omega, H), k \le m \right \},
$$
\item[(iii)] 
$$
A(\Omega, H) \equiv \{ f \in L^2(\Omega, H)| \bar \p f = 0\}, 
$$
\end{itemize}
where the anti-holomorphic derivative is understood in the sense of distributions. 
\end{definition}

Then it is known that the space $W^m_{\bar \p} (\Omega, H)$ is a Hilbert space with respect to the norm 
$$
||f||^2_{W^m_{\bar \p}} \equiv \sum_{k=0}^m ||\bar{\p}^k f ||^2_{L^2(\Omega)},
$$
and the spaces defined at (i), (iii) become, when $H = \mathbb{C}$, the usual spaces of square-summable and analytic functions, respectively. In this case, an estimate for the best approximation for $\bar z$ can be found in terms of the $L^2-$norm on $\Omega$, namely: 
$$
\inf_{\varphi \in A(\Omega, H)} ||\bar z - \varphi(z) ||_{L^2(\Omega)} \le  C_{\Omega} ||T_z||_{L^2(\Omega)}, 
$$
with $C_{\Omega}$ a constant determined entirely by the geometric data of $\Omega$. This follows directly from Proposition 2.2. in \cite{MP}, by setting $f = \bar z$. However, this approach does not yield lower-bound estimates for the quantization parameter.  

\subsection{Lower bounds for quantization parameter in the general case}

The topological bound Theorem~\ref{th1} is a consequence of a more general result which will be presented in full elsewhere, and is stated here without proof: 

\begin{claim}
Geometric quantization and algebraic deformation quantization on compact spaces are equivalent, for any choice of functional model.
\end{claim}

In essence, the main reason behind topological bounds for the possible values of the quantization parameter, in the sense discussed here, follows from the combination of restrictions imposed by the compatibility of simultaneous quantization on the boundary components $\Gamma_k$; consider, for instance, the case in which each curve $\Gamma_k$ corresponds to a Riemann surface $\Sigma_k$, and fundamental group $\pi_k$. The problem then becomes  equivalent to (due to the covariance of quadratic differentials) finding a non-trivial element in the intersection of the Teichm\"uler spaces of $\Sigma_k/\pi_k$, and, by the embedding of the universal  Teichm\"uler  space into the space of bounded holomorphic functions on $\mathbb{D}$, with $C^{\infty}$ norm, to the computation of the $C^{\infty}$ estimate 
$$
\inf_{\varphi \in X} ||\bar z - \varphi(z)|| \ge \inf_{\varphi \in A(\Omega)} ||\bar z - \varphi(z)||, 
$$
where $X$ is strictly included in $A(\Omega)$. Finally, if the functional model for the quantization is a distribution space of the type discussed in Definition~\eqref{de}, then the estimate for $\inf_{\varphi \in A(\Omega, H)} ||\bar z - \varphi(z) ||_{L^2(\Omega)}$ is greater than the minimal value $\lambda_m$ by usual norm inequalities. 

\section*{Appendix 1: The doubly-connected case: annular domains} \la{DP}

\begin{lemma} \la{l1}
Let $\Omega$ be a doubly-connected extremal domain with analytic boundary $\Gamma = \Gamma_1 \cup \Gamma_2.$ If $\varphi$ is the best analytic approximation to $\bar{z}$ in the supremum norm,     
\be \la{cm}
        \varphi'(z) = C [(\log h(z))']^2,
    \ee
with $h$ the conformal map from $\Omega$ to an annular domain $R_2 < |w| < R_1$, $C = $const.
\end{lemma}
\begin{proof}
When the connectivity $n=2$, a direct application of \cite[Theorem 1]{Jenkins} and covariance of quadratic differentials gives the desired result.
\end{proof}

\begin{lemma} \la{l2}
The diffeomorphism  $\mu: \Gamma_2 \to \Gamma_1$, defined through
$$
\mu(z) = h^{-1}\left (\frac{R_1}{R_2}h(z) \right ),
$$
is a M\"obius transformation.
\end{lemma}
\begin{proof}
Clearly, $\mu$ is a diffeomorphism by composition law. 
By definition, for any $z_2 \in \Gamma_2, z_1 = \mu(z_2) \in \Gamma_1$,
\be
\frac{h(z_1)}{h(z_2)} = \frac{R_1}{R_2}, \quad (h \circ \mu)(z_2) = \frac{R_1}{R_2}h(z_2),
\ee
so the chain rule and Lemma~\ref{l1} give
\be \la{rr}
h'(z_1)\cdot \mu'(z_2) = \frac{R_1}{R_2}h'(z_2) \Rightarrow \mu'(z_2) = \frac{h(z_1)}{h(z_2)}  \cdot \frac{h'(z_2)}{h'(z_1)} = \sqrt{\frac{\varphi'(z_2)}{\varphi'(z_1)}}
\ee
Therefore,
\be \la{invar}
\frac{dz_1}{dz_2} = \sqrt{\frac{\varphi'(z_2)}{\varphi'(z_1)}} \Rightarrow \varphi'(z_2)dz_2^2 = \varphi'(z_1)dz_1^2.
\ee
Since $\varphi'(z)$ is a quadratic differential (c.f. Theorem~\ref{function}, as a Schwarzian), it transforms under composition with the map $\mu(z)$ as
\be
\varphi'(z_2)dz_2^2 = \varphi'(z_1)dz_1^2 + \mathcal{S}(\mu(z_1))dz_1^2,
\ee
with $\mathcal{S}(\mu(z))$ the Schwarzian of the map $\mu(z)$. Thus, Eq.~\ref{invar} gives $\mathcal{S}(\mu) \equiv 0$, so
$\mu(z)$ is a M\"obius transformation.
\end{proof}

\begin{lemma} \la{l3}
Under the mapping $\mu(z)$, the curvatures $\kappa_j$ of $\Gamma_j$ at $z_j$ ($j = 1, 2$), $z_1 = \mu(z_2)$, are related via
\be \la{id}
(1- \lambda \kappa_1) ds_1^2 
= (1+\lambda \kappa_2) ds_2^2,
\ee
and
\be \la{id2}
\quad \frac{d \kappa_1}{ds_2} = \frac{d\kappa_2}{d s_1}.
\ee
\end{lemma}

\begin{proof}
Eq.~\ref{id} follows from Eq.~\ref{l0} and Eq.~\ref{invar}. Now using the general formula valid for any curve diffeomorphism
$$
S_1(z) = \mu^{\sharp}(S_2(\mu^{-1}(z))),
$$
and the fact that $\mu^{-1}, \mu^{\sharp}$ are also M\"obius transformations, we conclude that the Schwarzians of the Schwarz functions $S_{1,2}$ are related via
$$
\mathcal{S}(S_1) = \mathcal{S}(\mu^{\sharp}\circ S_2 \circ \mu^{-1}) = \mathcal{S}(S_2)[(\mu^{-1})']^2.
$$
since Schwarzians are invariant under left composition and covariant under right composition with M\"obius transformations.

Evaluating the Schwarzian of the Schwarz function of any smooth curve gives the general expression
$$
\mathcal{S}(S(z)) = i \bar \tau^2 \frac{d \kappa}{ds},
$$
so we arrive at the identity
\be \la{eq}
\bar{\tau}^2(z_1) \frac{d\kappa}{ds}(z_1) = \bar{\tau}^2(z_2) \frac{d\kappa}{ds}(z_2)\left ( \frac{dz_2}{dz_1}\right )^2, \quad (\forall) z_2 \in \Gamma_2,
\ee
or
\be
\frac{d\kappa_1}{ds_1} = \frac{d\kappa_2}{ds_2} \left ( \frac{ds_2}{ds_1} \right )^2,
\ee
which implies Eq.~\ref{id2}.
\end{proof}

\begin{thm} \la{final}
Under the conditions of Lemmas~\ref{l1}-\ref{l3}, the curves $\Gamma_{1,2}$ are concentric circles.
\end{thm}

\begin{proof}
From Lemma~\ref{l2}, there are 3 complex numbers $a, b, c$ such that
\be \la{mob}
z_1 = a + \frac{b}{z_2-c}.
\ee
Introducing the notation $r_1e^{i \phi_1} = (z_1 -a), r_2e^{i \phi_2} = (z_2 -c)$, we obtain
\be \la{eqa}
r_1\cdot r_2 = |b|, \quad \phi_1 + \phi_2 = \mbox{const.}
\ee
By differentiation of Eq.~\ref{mob}, we have
\be \la{area1}
\frac{dz_1}{dz_2} = -\frac{z_1-a}{z_2-c}, \quad
\frac{ds_1}{ds_2} = \frac{|z_1-a|}{|z_2-c|} = \frac{r_1}{r_2}. 
\ee
Differentiating Eq.~\ref{eqa},
\be \la{ew1}
\frac{1}{r_1}\frac{dr_1}{d\phi_1} = \frac{1}{r_2}\frac{dr_2}{d\phi_2}, \quad
\frac{1}{r_1}\frac{d^2r_1}{d\phi_1^2} - \left ( \frac{1}{r_1}\frac{dr_1}{d\phi_1} \right)^2
= - \left [ \frac{1}{r_2}\frac{d^2r_2}{d\phi_2^2} - \left ( \frac{1}{r_2}\frac{dr_2}{d\phi_2} \right)^2 \right ].
\ee
Since for any planar curve $r(\phi)$ the extrinsic curvature can be expressed as
\be \la{ew2}
\kappa = \frac{1}{r}\cdot \frac{1 + \frac{1}{r}\frac{d^2r}{d\phi^2} - 2\left ( \frac{1}{r}\frac{dr}{d\phi} \right)^2}
{\left [ 1 + \left ( \frac{1}{r}\frac{dr}{d\phi} \right)^2 \right ]^{3/2}} =
\frac{1}{r}\cdot \frac{1 - \left ( \frac{1}{r}\frac{dr}{d\phi} \right)^2}
{\left [ 1 + \left ( \frac{1}{r}\frac{dr}{d\phi} \right)^2 \right ]^{3/2}} + \frac{1}{r}\cdot \frac{\frac{1}{r}\frac{d^2r}{d\phi^2} - \left ( \frac{1}{r}\frac{dr}{d\phi} \right)^2}
{\left [ 1 + \left ( \frac{1}{r}\frac{dr}{d\phi} \right)^2 \right ]^{3/2}},
\ee
we obtain from Eqs.~\ref{ew1}-\ref{ew2} that
\be \la{dif}
\kappa_1r_1 - \kappa_2r_2 = 2\cdot \frac{\frac{1}{r_1}\frac{d^2r}{d\phi_1^2} - \left ( \frac{1}{r_1}\frac{dr_1}{d\phi_1} \right)^2}
{\left [ 1 + \left ( \frac{1}{r_1}\frac{dr_1}{d\phi_1} \right)^2 \right ]^{3/2}} =
2 \cdot \frac{\frac{d}{d\phi_1} \left ( \frac{1}{r_1}\frac{dr_1}{d\phi_1}\right )}
{\left [ 1 + \left ( \frac{1}{r_1}\frac{dr_1}{d\phi_1} \right)^2 \right ]^{3/2}}
\ee
Now Eq.~\ref{id} and Eq.~\ref{area1} give
\be
1-\lambda \kappa_1 = (1 + \lambda \kappa_2)\left( \frac{r_2}{r_1} \right)^2,
\ee
so by differentiating with respect to $s_2$, we arrive using Eq.~\ref{area1} at
\be
-\lambda \frac{d \kappa_1}{ds_2} = \left [ \frac{d}{ds_2}\left ( \frac{r_2}{r_1} \right )^2 \right ] (1+\lambda \kappa_2) + \lambda \frac{d\kappa_2}{ds_2}\cdot \left( \frac{d s_2}{d s_1} \right)^2,
\ee
so from Eq.~\ref{id} and elementary manipulations,
\be
-2\lambda \frac{d \kappa_1}{ds_2} = \left [ \frac{d}{ds_2}\log \left ( \frac{r_2}{r_1} \right )^2 \right ]\cdot (1+\lambda \kappa_2) \left ( \frac{r_2}{r_1} \right )^2
\ee
Applying Eq.~\ref{id2} again, we finally arrive at
\be
\frac{d}{ds_2} \log(1-\lambda \kappa_1) =  \frac{d}{ds_2}\log \left ( \frac{r_2}{r_1} \right ),
\ee
or
\be \la{rw}
1-\lambda \kappa_1 = K\cdot \frac{r_2}{r_1} = K \frac{ds_2}{ds_1} = (1+\lambda \kappa_2)  \left (\frac{ds_2}{ds_1}\right)^2,
\quad K = \mbox{ const.}
\ee
Multiplying Eq.~\ref{rw} by $ds_1$ and integrating over $\Gamma_1$, we get
\be \la{q1}
\oint_{\Gamma_1}(1-\lambda \kappa_1) ds_1 = K\oint_{\Gamma_2} ds_2 \Rightarrow 2\pi \lambda = L_1 - K\cdot L_2.
\ee
Likewise, Eq.~\ref{rw} gives $K ds_1 = (1+\lambda \kappa_2) ds_2$, so integrating over $\Gamma_2$, we have
\be \la{q2}
K\cdot L_1 = L_2 + 2\pi \lambda.
\ee
Consistency of Eqs.~\ref{q1}, \ref{q2} imposes
$$
(L_1 + L_2)(K-1) = 0 \Rightarrow K=1.
$$
Thus, Eq.~\ref{rw} becomes
\be \la{oh}
(1-\lambda \kappa_1)r_1 = r_2, \quad (1+\lambda \kappa_2) r_2 = r_1.
\ee
Adding the two identities in \eqref{oh}, we obtain
\be \la{po}
\kappa_2 r_2 - \kappa_1 r_1 = 0.
\ee
Finally, Eq.~\ref{po} and Eq.~\ref{dif} lead to
\be
\frac{d}{d\phi_1} \left ( \frac{1}{r_1}\frac{dr_1}{d\phi_1}\right ) = 0 \Rightarrow r_1(\phi_1) = A\cdot e^{\alpha \phi_1},
\ee
with $A, \alpha$ constants. Therefore, $\Gamma_1$ is either a circle ($\alpha = 0$), or a spiral ($\alpha \ne 0$), and being a closed curve, it can only be a circle of radius $r_1 = R_1$. From Eqs.~\ref{start1}-\ref{start2} and the Schwarz function for a circle, it follows that $\Gamma_2$ is also a circle concentric  with $\Gamma_1$, of radius $r_2 = R_2$, and $\lambda = R_1 - R_2$.
\end{proof}

\section{Acknowledgments}

A portion of this work was developed as part of a different project, concerning the classical  extremal problem of best approximations in the space of analytic functions. The author is grateful to Ar. Abanov, C. Beneteau, B. Gustavsson, D. Khavinson, M. Putinar, and A. Vasiliev for useful discussions and feed-back, especially regarding the calculations contained in Section 3 and the Appendix. 


 \bibliographystyle{amsplain}

\end{document}